\documentclass[sigconf]{aamas} 

\usepackage{balance} 
\usepackage{amsfonts}

\setcopyright{ifaamas}
\copyrightyear{2024}
\acmYear{2024}
\acmDOI{}
\acmPrice{}
\acmISBN{}

\usepackage{amsmath}
\usepackage{relsize}
\usepackage{enumerate}

\title[AAMAS-2024 Formatting Instructions]{Majority-based Preference Diffusion on Social Networks}

\author{Ahad N. Zehmakan}
\affiliation{
  \institution{The Australian National University}
  \city{}
  \country{}}
\email{ahadn.zehmakan@anu.edu.au}

\begin{abstract}
We study a majority based preference diffusion model in which the members of a social network update their preferences based on those of their connections. Consider an undirected graph where each node has a strict linear order over a set of $\alpha$ alternatives. At each round, a node randomly selects two adjacent alternatives and updates their relative order with the majority view of its neighbors.

We bound the convergence time of the process in terms of the number of nodes/edges and $\alpha$. Furthermore, we study the minimum cost to ensure that a desired alternative will ``win'' the process, where occupying each position in a preference order of a node has a cost. We prove tight bounds on the minimum cost for general graphs and graphs with strong expansion properties.

Furthermore, we investigate a more light-weight process where each node chooses one of its neighbors uniformly at random and copies its order fully with some fixed probability and remains unchanged otherwise. We characterize the convergence properties of this process, namely convergence time and stable states, using Martingale and reversible Markov chain analysis.

Finally, we present the outcomes of our experiments conducted on different synthetic random graph models and graph data from online social platforms. These experiments not only support our theoretical findings, but also shed some light on some other fundamental problems, such as designing powerful countermeasures.
\end{abstract}

\newcommand{\BibTeX}{\rm B\kern-.05em{\sc i\kern-.025em b}\kern-.08em\TeX}

\begin{document}


\pagestyle{fancy}
\fancyhead{}


\maketitle 


\section{Introduction}
Humans constantly form and update their preferences on different topics. In the process of making such decisions, we tend to rely not only on our own personal judgment and knowledge, but also that of others, especially those whose opinion we value and trust. As a result, opinion diffusion, influence propagation, and (mis)-information spreading can affect different aspects of our lives from economy and defense to fashion and personal affairs.

Recent years have witnessed a booming development of online social networking platforms. The enormous popularity of these platforms has led to fundamental changes in how humans share and form opinions. Social phenomena such as disagreement and polarization that have existed in human societies for millennia, are now taking place in an online virtual world and are tightly woven into everyday life, with a substantial impact on society.

The motive to gain insights on how opinions are shaped and evolved in multi-agent systems has been the driving force behind an interdisciplinary research effort in diverse areas such as sociology~\cite{moussaid2013social}, economics~\cite{jackson2011overview}, political science~\cite{marsh1999state}, mathematics~\cite{balogh2012sharp}, physics~\cite{galam2008sociophysics, n2020rumor}, and computer science~\cite{chistikov2020convergence, liu2023fast}. Within the field of computer science, especially computational social choice and algorithmic game theory, there has been a rising interest in developing and analyzing mathematical models which simulate the opinion diffusion in a network of individuals, cf.~\cite{bredereck2017manipulating, out2021majority, li2023graph}. Of course, in reality the opinion formation processes are too complex to be expressed in purely theoretical terms, but the goal is to shed some light on their general principles, which are otherwise hidden by the convolution of the full phenomenon.

The following generic and abstract model is the foundation for most of the proposed opinion diffusion models. Consider a graph where each node has a color and the nodes change their color according to an updating rule, which is a function of the color of the neighboring nodes, in a defined order. The graph is meant to represent a social network, where the individuals are modeled as nodes and edges indicate relations between them, e.g., friendship, common interests, or advice. The color of a node stands for its preference about a certain topic, e.g., an order over a set of candidates.

In the abundance of opinion diffusion models, cf.~\cite{zehmakan2023random, li2012phase, chistikov2020convergence, zhu2022nearly, zhou2021maximizing, liu2018active, okawa2022predicting, valentini2014self}, the majority-based models, where an individual updates its opinion to reflect the majority view among its connections, have gained substantial popularity and interest. In most of these models, the individual opinions are modelled as binary views on a given issue, that is, a node has one of the two colors black and white. However, opinions are sometimes complex objects that cannot be accurately modelled in binary terms. As a result, there is an emerging line of research, cf.~\cite{brill2016pairwise,hassanzadeh2013building, brill2018pairwise}, which studies the models where the preferences are expressed as linear orders over a set of alternatives, following standard conventions in voting theory, cf.~\cite{taylor2005social, rossi2011short}.
In the present paper, we focus on a model, called the Preference Diffusion model, which falls under the umbrella of the above line of research, that is, the nodes have an ordered list of preferences and follow a majority updating rule. This model was proposed in~\cite{hassanzadeh2013building} and was generalized by~\cite{brill2016pairwise}.

Arguably, the most well-studied problem in the area of opinion diffusion models is determining the convergence properties of the opinion dynamics: Is convergence to stable states guaranteed and if yes, what are the upper and lower bounds on the convergence time? Since in most cases the opinion dynamics can be modelled as a Markov process, this problem is usually equivalent to determining the stationary distribution and convergence time of the corresponding Markov chain, cf.~\cite{frischknecht2013convergence}.

Marketing campaigns routinely use online social platforms to sway people’s opinions in their favor, for instance by targeting segments of members with free or discount samples of their products or misleading information. The ultimate goal is to trigger a large cascade of product adoptions through the word-of-mouth effects by targeting a small set of influential individuals.
Therefore, a question which has been studied extensively is finding the minimum cost of manipulating the preferences of the individuals to ensure that a favored alternative dominates once the process ends, cf.~\cite{avin2019majority, faliszewski2022opinion}.

Our main focus is to address the above problems for the Preference Diffusion model by providing theoretical findings (building on various graph and probabilistic tools) and conducting experiments on data from real-world social networks.


\vspace{-0.3cm}
\subsection{Preliminaries}
\label{preliminaries}
\textbf{Graph Definitions.} Let $G=\left(V,E\right)$ be a simple undirected graph, where $V:=\{v_1,\cdots,v_n\}$ and $E\subseteq \{\{v_i,v_j\}: v_i,v_j\in V, i\ne j\}$. (We always assume that $G$ is undirected; otherwise, it is mentioned explicitly.) Furthermore, $n$ and $m$ denote the number of nodes and edges in $G$, respectively.

For a node $v\in V$, $\Gamma\left(v\right):=\{u\in V: \{u,v\} \in E\}$ is the \emph{neighborhood} of $v$. For a set $S\subset V$, we define $\Gamma_S\left(v\right):=\Gamma\left(v\right)\cap S$. Moreover, $d\left(v\right):=|\Gamma\left(v\right)|$ is the \emph{degree} of $v$. For two node sets $S$ and $S'$, we define $e\left(S,S'\right):=|\{\left(v,u\right)\in S\times S': \{v,u\}\in E\}|$ where $S\times S'$ is the Cartesian product of $S$ and $S'$.

\textbf{Random Graph.} Let $\mathcal{G}_{n,q}$ denote the Erd\H{o}s-R\'{e}nyi (ER) random graph, which is the random graph on the set $\{v_1,\cdots,v_n\}$ where each edge is present independently with probability (w.p.) $q$.

\textbf{Expansion.} There exist different parameters to measure how expansive (i.e., well-connected) a graph is. We consider an algebraic characterization of expansion. Assume that $A\left(G\right)$ is the adjacency matrix of graph $G=\left(V,E\right)$ and define $D$ to be the diagonal matrix where the entries of the diagonal are the degrees of the nodes. We consider the normalized adjacency matrix $M=D^{-\frac{1}{2}}AD^{-\frac{1}{2}}$, which is symmetric. Let $1=\lambda_1\ge \lambda_2\ge\cdots\ge\lambda_n\ge-1$ be the eigenvalues of $M$. We denote the second-largest absolute eigenvalue of the normalized adjacency matrix by $\lambda\left(G\right):=\max_{2\le i\le n}|\lambda_i|$. For our purpose here, it suffices to note that graph $G$ has stronger expansion properties (i.e., is more well-connected) when $\lambda\left(G\right)$ is smaller.

\textbf{Preference Diffusion Model.} Let $\mathcal{A}$ be a set of $\alpha$ \emph{alternatives} for some $\alpha\in \mathbb{N}$. The preferences of a node $v_i\in V$ are represented as a \emph{strict linear order} $\succ_i\subset \mathcal{A}\times\mathcal{A}$. For two distinct alternatives $a, b\in \mathcal{A}$, we write $a\succ_i b$ if $(a,b)\in \succ_i$, i.e., $a$ is \emph{preferred} over $b$. We always assume that an order is a strict linear order; otherwise, it is mentioned explicitly. Thus, we always have either $a\succ_i b$ or $b \succ_i a$. We say the alternatives $a,b$ are \emph{adjacent} in $\succ_i$ if they appear next to each other in the order, i.e., there is no alternative $c$ such that $a\succ_i c\succ_i b$ or $b\succ_i c\succ_i a$. Furthermore, an alternative $a$ is in the $k$-th \emph{position} in order $\succ_i$ if exactly $k-1$ other alternatives are preferred over $a$. For example, let $\mathcal{A}=\{a,b,c\}$ (which implies $\alpha=3$) and $\succ_1=\{(a,b), (b,c), (a,c)\}$ for node $v_1\in V$. This can also be written as $a\succ_1 b\succ_1 c$, where $a$, $b$, and $c$ are in the first, second, and third position, respectively. In this order, $a$ and $b$ (similarly, $b$ and $c$) are adjacent. Furthermore, $a$ is preferred over $b$ and $b$ is preferred over $c$.

A \emph{profile} $\mathcal{P}=(\succ_1,\cdots, \succ_n)$ contains the preferences of all nodes. For two alternatives $a,b$ and a profile $\mathcal{P}$, let $V_{ab}^{\mathcal{P}}$ denote the set of nodes where $a$ is preferred over $b$ in $\mathcal{P}$ and let $n_{ab}^{\mathcal{P}}:=| V_{ab}^{\mathcal{P}}|$. We sometimes write $n_{ab}$, when $\mathcal{P}$ is clear from the context. Let $V_{\succ}^{\mathcal{P}}$ be the set of nodes which have order $\succ$ in profile $\mathcal{P}$ and $\overline{V}_{\succ}^{\mathcal{P}}:=V\setminus V_{\succ}^{\mathcal{P}}$. Furthermore, we define $n_{\succ}^{\mathcal{P}}:= |V_{\succ}^{\mathcal{P}}|$ and $\bar{n}_{\succ}^{\mathcal{P}}:=|\overline{V}_{\succ}^{\mathcal{P}}|$.

Consider a graph $G$ and an initial profile $\mathcal{P}_0$. In the Preference Diffusion (PD) model, in each round, some nodes are selected to update their preferences. Each selected node $v_i$ chooses two distinct alternatives in $\mathcal{A}$ uniformly at random (u.a.r.), say $a,b$. If $a,b$ are adjacent in $\succ_i$ and more than half of $v_i$'s neighbors disagree with the relative order of $a,b$ in $\succ_i$, then it swaps $a,b$ in $\succ_i$; otherwise, it remains unchanged. In the Synchronous Preference Diffusion (SPD) model, the set of nodes which update is equal to $V$ in every round. In the Asynchronous Preference Diffusion (APD) model, in each round a node is chosen u.a.r. to update.

We also define the Random Preference Diffusion model, where starting from an initial profile $\mathcal{P}_0$, in each round every node decides independently and with some fixed probability $0<q<1$ to update its order or not. If a node decides to update its order, it picks one of its neighbors u.a.r. and copies its order fully.

We let $\mathcal{P}_t$ for $t\in \mathbb{N}$ denote the profile obtained in round $t$, where graph $G$, initial profile $\mathcal{P}_0$, and the model are clear from the context. An update from $\mathcal{P}_t$ to $\mathcal{P}_{t+1}$ is \emph{effective} if $\mathcal{P}_t\ne \mathcal{P}_{t+1}$. A profile $\mathcal{P}$ is called \emph{fixed} if there is no effective update possible from $\mathcal{P}$.

We should mention that the PD model in the special case of $|\mathcal{A}|=2$ is usually known as the Majority model, which is a very well-studied model in the literature, cf.~\cite{frischknecht2013convergence, zehmakan2020opinion}. Thus, the PD model can be seen as a generalization of the Majority model.

\textbf{Winning.} In Section~\ref{convergence}, we prove that the PD process always reaches a fixed profile. We say an alternative $a$ \emph{wins} the process if it is in the first position for all nodes in the final profile. Similarly, for an order $\succ$, we say it wins if all nodes have order $\succ$ in the final profile.
It is worth mentioning that most of our results would still hold if we relax the definition of winning, e.g., we require that $80\%$ of nodes satisfy the desired property instead of all nodes.

\textbf{Cost.} For a graph $G$, let a \emph{placement} determine a position from $1$ to $\alpha$ on each node. A placement is called a \emph{solution} in the APD (or SPD) model on $G$ whenever the following holds: If an alternative occupies the positions determined by the placement in the initial profile, then that alternative wins, regardless of how other alternatives are positioned and the random choices of the process. (Most of our results would also hold for a relaxed version of this definition.) In a solution, the \emph{cost} for each node $v$ is defined to be $\alpha$ minus the selected position in $v$ (for example, the cost is $\alpha-1$ if the first position is chosen and $0$ if the last one). The cost of a solution is the summation of the cost of all nodes. We define $\mathcal{MC}(G)$ to be the minimum cost for a solution on $G$. Furthermore, we let $\mathcal{NS}(G)$ denote the number of solutions for $G$.

\textbf{Condorcet winner.} We say an alternative $a$ is \emph{$\epsilon$-Condorcet} in a profile $\mathcal{P}$ and for some constant $\epsilon>0$ if $n_{ab}^{\mathcal{P}}> n_{ba}^{\mathcal{P}}+\epsilon n$ for any alternative $b\ne a$. (In voting theory~\cite{taylor2005social, rossi2011short}, this is usually known as the Condorcet winner for $\epsilon=0$, where $a$ wins against every other alternative in a head-to-head comparison.) Similarly, we say an order $\succ$ is \emph{$\epsilon$-Condorcet} in a profile $\mathcal{P}$ when the following holds: For every two alternatives $a,b\in \mathcal{A}$, if $a\succ b$, then $n_{ab}^{\mathcal{P}}>n_{ba}^{\mathcal{P}}+\epsilon n$.

\textbf{Some Useful Inequalities.} Let us provide the Chernoff bound and Chebyshev's inequality, cf.~\cite{dubhashi2009concentration}, which come in handy later.

\begin{itemize}
\item \textit{Chernoff Bound:} Suppose $x_1,\cdots,x_n$ are independent Bernoulli random variables and let $X$ denote their sum. Then, $\mathbb{P}[X\leq \left(1-\epsilon'\right)\mathbb{E}[X]]\leq \exp\left({-\frac{\epsilon'^2\mathbb{E}[X]}{2}}\right)$ for $0\leq \epsilon'\leq 1$.
\item \textit{Chebyshev's Inequality:} Let $X$ be a random variable with finite variance and $f>0$. Then, we have $\mathbb{P}[|X-\mathbb{E}[X]|\ge f]\le Var[X]/f^2$.
\end{itemize}

\textbf{Assumptions.} All logarithms are to base $e$ and we let $n$ (the number of
nodes) tend to infinity. We say an event happens with high probability (w.h.p.) when it occurs w.p. $1-o(1)$. As mentioned, the PD model is equivalent to the Majority model when $|\mathcal{A}|=2$, but in the present paper we focus on the case of $|\mathcal{A}|>2$. Thus, whenever we talk about the PD model, it is assumed that $|\mathcal{A}|>2$; otherwise, we use the term Majority model. (Note that one of the main difficulties our proofs need to overcome is that when a node updates, it compares two \textit{randomly} selected alternatives. For $|\mathcal{A}|=2$, a node deterministically chooses the only two existing alternatives. Most of our proofs can easily be adjusted to cover the case of $|\mathcal{A}|=2$, but this is left out since it usually requires a case distinction, which would disrupt the flow of the proof without adding any technical value, and as we explain in Section~\ref{related} some of such results are already known for $|\mathcal{A}|=2$.) Furthermore, we always assume that the underlying graph is connected; otherwise, it is stated explicitly.

\subsection{Our Contribution}
\label{contribution}

We study the PD model developed in~\cite{brill2016pairwise}, based on a preliminary version of the model proposed in~\cite{hassanzadeh2013building}. However, unlike the present paper which focuses on undirected graphs, in~\cite{brill2016pairwise} the process was studied on directed graphs, especially acyclic ones and in~\cite{hassanzadeh2013building} the special case of a complete graph was analyzed.

How long does it take for the process to reach a stable state? And what does such a stable state look like? Building on some potential function arguments, we prove that the PD process always converges to a fixed profile. For the asynchronous set-up, we prove that this happens in $\mathcal{O}(nm\alpha^4)$ rounds w.h.p. We also argue the tightness of this bound.

We study the minimum cost to guarantee that an alternative $a$ wins the process, i.e., $\mathcal{MC}(G)$. We prove that, in both the SPD and APD model, the minimum cost is at least $\sqrt{n}(\alpha-1)$ and this bound is tight, that is, there are graphs where there is a solution of cost $\sqrt{n}(\alpha-1)$.

The above result implies that there are graph classes where an extremely small subset of nodes, i.e., $\sqrt{n}(\alpha-1)$ nodes, has a disproportionate amount of power and can engineer the outcome of the process. A natural question arises is whether there are graphs where the power is distributed more uniformly among the nodes, which as a result would limit the potency of an adversary who attempts to manipulate the process. In~\cite{hassanzadeh2013building}, it was shown that when the underlying graph is complete, then the minimum cost to guarantee a win is $n(\alpha-1)/2$, that is, an adversary needs to bribe at least half of the nodes to ensure its desired outcome. Thus, a complete graph structure is immensely resilient against the mentioned adversarial attacks. However, demanding the graph to be complete is quite restrictive and also unrealistic. We prove a similar result for a much larger class of graphs, namely graphs with a certain level of expansion and regularity. We actually provide our results in a more general form and in terms of $\epsilon$-Condorcet with a logarithmic bound on the convergence time (which is shown to be tight).


We also initiate the study of the number of solutions for a graph $G$. We provide tight bounds for different classes of graphs. In particular, for a cycle $C_n$ we prove the bound of $\mathcal{NS}(C_n)=\Tilde{\Theta}(\psi^n)$ for some $\psi\in(\alpha^{1/3}, \alpha^{1/3}+0.22)$, where $\Tilde{\Theta}$ hides polynomial terms in $n$ and $\alpha$. For the proof, we show that the number of solutions can be bounded by a recurrence relation, which can then be solved using classical methods.

The primary objective of introducing the PD model in~\cite{hassanzadeh2013building} was to develop a method for reaching consensus on the choice of order in the network. However, as it is observed in~\cite{hassanzadeh2013building} and~\cite{brill2016pairwise}, the PD model does not fulfill this objective for many classes of graphs. On the other hand, we prove that the Random PD model always reaches consensus and this takes $\mathcal{O}(n^4)$ rounds in expectation. For graphs with strong expansion properties, the stronger bound of $\mathcal{O}(n\log n)$ exists. Furthermore, the Random PD model is a ``lightweight'' and ``fair'' process. It is lightweight since it polls the order of only one of its neighbors, in contrast to the APD and SPD model which require the full knowledge of the preferences of the neighboring nodes. It is fair in the sense that, as we will prove, the probability that the process converges to an order $\succ$ is proportional to the summation of the degree of nodes which hold order $\succ$ initially.

We conduct several experiments for our models on different classes of graphs, such as cycles and ER random graphs, and graph data from real-world social networks, such as Facebook and Twitter. Our experiments support and complement our theoretical findings, such as the results on the convergence time and resilience against adversarial attacks. Furthermore, we develop and evaluate two countermeasures to stop an adversary from engineering the outcome of the process. Roughly speaking, the first countermeasure requires each node to choose some of its connections at random and in the second countermeasure, the nodes give more weight to the preferences of the nodes which are more ``similar'' to them.

\vspace{-0.3cm}
\subsection{Related Work}
\label{related}

As mentioned, our main focus is on the PD model, which was introduced in~\cite{hassanzadeh2013building, brill2016pairwise}. In~\cite{hassanzadeh2013building}, the PD model was studied for the special case of complete graphs and in~\cite{brill2016pairwise} the study of the PD model was extended to general directed graphs, with a special focus on acyclic graphs. Other similar preference diffusion models have also been considered. For example, in~\cite{brill2018pairwise} a variant called liquid democracy was studied, where voters have partial preference orders and can delegate their vote to another voter of their choice for some preference comparisons. Furthermore, a majority based model was introduced in~\cite{botan2017propositionwise}, where each node has to choose a fixed number of alternatives from a pool of alternatives.

\textit{What does the final state look like?} In~\cite{hassanzadeh2013building}, it was proven that for the APD model on a complete graph, the process always reaches a fixed profile and if there is a Condorcet winner $a$ in the initial profile, $a$ wins the process. Brill et al.~\cite{brill2016pairwise} proved that for both the APD and SPD model on a directed acyclic graph, the corresponding Markov chain is also acyclic, which implies that the process always reaches a fixed profile. For general graphs, they showed that if the initial profile and all profiles which possibly arise during the process satisfy a certain property, then the process reaches a fixed profile. However, it was left open whether this is true for any graph regardless of the initial profile. They also characterized the set of possible final profiles for acyclic directed graphs and simple cycles.

For the Majority model, i.e., the PD model with $|\mathcal{A}|=2$, it is known~\cite{GOLES1980187} that the asynchronous version always reaches a fixed profile while the synchronous variant might deterministically switch between two profiles.

\textit{What is the convergence time?} To our knowledge, we are the first to study the convergence time of the APD and SPD model. However, the convergence time has been studied extensively for other opinion diffusion models. Poljak and Turzik~\cite{poljak1986pre} proved the upper bound of $\mathcal{O}(n^2)$ on the convergence time of the Majority model (which is shown~\cite{frischknecht2013convergence} to be tight, up to some poly-logarithmic factor) while stronger bounds are
known for special classes of graphs. For example, for a $d$-regular random graph, the tight bound of $\mathcal{O}(\log_d n)$ was proven in~\cite{gartner2018majority}. On the other hand, once some randomness is introduced to the updating rule, the process might need exponentially many rounds to converge, cf.~\cite{lesfari2022biased}.


\textit{What is the minimum cost to win?} The minimum cost to guarantee the dominance of a particular alternative in the final profile has been investigated for different majority based models. In~\cite{zehmakan2020opinion}, this problem is studied for the Majority model on expander graphs, and in~\cite{avin2019majority} the problem is investigated experimentally in the set-up where a certain subset of nodes, known as elites, have a higher influence factor than ordinary nodes. Furthermore, Zehmakan~\cite{zehmakan2021majority} focused on a variant where initially all nodes are neutral about the alternatives except a small subset of nodes, known as seeds or early adopters, which will lead the preference formation process.

In addition to bounding the minimum cost for different classes of graphs, the problem has been also studied from an algorithmic perspective. For the Majority model, it is known~\cite{mishra2002hardness, chen2009approximability} the problem of finding the minimum cost of a solution for a given graph is NP-hard. Similar complexity results are provided for different variants of the PD model, cf.~\cite{faliszewski2022opinion, elkind2009swap}. These works also propose some (approximation) algorithms and evaluate their performance both theoretically and experimentally.

\section{Convergence Time}
\label{convergence-time-section}
We provide our results about the convergence of the APD model in Theorem~\ref{convergence-time}. Its proof is built on Lemma~\ref{coin-flip} whose proof is given in the appendix, Section~\ref{coin-flip-appendix}. We discuss the tightness of our convergence results in Section~\ref{convergence-time-tightness} by giving an explicit construction.
\label{convergence}
\begin{lemma}
\label{coin-flip}
In the \emph{coin flip} process with parameters $p$ and $K$, for some probability $0<p\le 1$ and a positive integer $K$, we keep flipping a coin, which comes head w.p. $p$ and tail otherwise independently, until we see $K$ heads. Let the random variable $X$ denote the number of flips required. Then, we have $\mathbb{P}\left[\frac{1}{2}(K/p)< X< \frac{3}{2}(K/p)\right]\ge 1- 1/(4K)$.
\end{lemma}

\begin{theorem}
\label{convergence-time}
In the APD model on a graph $G=(V,E)$, the process reaches a fixed profile and this takes at most $nm\alpha^4$ rounds w.h.p.
\end{theorem}
\begin{proof}
For an edge $e =\{v_i,v_j\}\in E$ and profile $\mathcal{P}$, let us define the potential $\phi_{e}(\mathcal{P})$ of the edge $e$ in $\mathcal{P}$ to be the number of pairs of two distinct alternatives $a, b$ such that $v_i$ and $v_j$ disagree on the order of $a$ and $b$, i.e., either $a\succ_i b$ and $b\succ_j a$ or $b\succ_i a$ and $a\succ_j b$. Let $\Phi_G(\mathcal{P}):= \sum_{e\in E}\phi_e(\mathcal{P})$ be the potential of graph $G$ in $\mathcal{P}$.

Note that the potential of an edge $e$ is at most ${\alpha \choose 2}$ since there are this many pairs of distinct alternatives. Thus, $\Phi_G(\mathcal{P})\le m{\alpha \choose 2}$ for any profile $\mathcal{P}$. It is also easy to see that $\Phi_G(\mathcal{P})\ge 0$. Furthermore, whenever we swap two adjacent alternatives $a,b$ in the order of a node $v_i$ in the APD model, the potential decreases at least by one. This is simply true because there is a subset of neighbors of size larger than $|\Gamma(v_i)|/2$ that $v_i$ disagreed with on the relative order of $a$ and $b$ before the update, but agrees with after. Thus, if the process is in a profile with at least one effective update, it eventually updates and the potential decreases. Since the initial potential is at most $m{\alpha \choose 2}$ and it can never become negative, the process must eventually reach a fixed profile.

It remains to prove the bound on the convergence time. If the process has not reached a fixed profile, there is at least a node $v_i$ and a pair of alternatives $a,b$ such that $a,b$ are adjacent in $\succ_i$ and $v_i$ disagrees with more than half of its neighbors on the order of $a,b$. If we swap $a,b$ in $\succ_i$, the potential function decreases at least by one. The probability that the node $v_i$ and the pair $a,b$ are chosen is equal to $1/(n{\alpha \choose 2})$. Thus, we start with a potential of at most $K=m{\alpha \choose 2}$ and in each round it decreases at least by one w.p. at least $p=1/(n{\alpha \choose 2})$, regardless of all the previous steps, and the process stops when it reaches potential zero (or possibly even earlier). Thus, the convergence time of the process can be upper-bounded by the number of flips required by the coin flip process, defined in Lemma~\ref{coin-flip}, with parameters $p$ and $K$. Let the random variable $X$ denote the required number of coin flips. By applying Lemma~\ref{coin-flip}, we have $\mathbb{P}\left[X< \frac{3}{2}(K/p)\right]\ge 1-(4/K)$. Since $(3/2)(K/p)=(3/2)nm{\alpha \choose 2}^2\le nm\alpha^4$ and $4/K\le 4/m=o(1)$ (note since we assume $G$ is connected, $m\ge n-1$), we conclude that the process reaches a fixed profile in at most $nm\alpha^4$ rounds w.h.p.
\end{proof}

\textbf{Synchronous Set-up.} The set of fixed profiles is identical in both the APD and SPD model. Furthermore, if a transition from profile $\mathcal{P}$ to $\mathcal{P}^{\prime}$ is possible in the APD model, it is also possible in the SPD model (because there is a non-zero probability that in the SPD model all nodes choose two non-adjacent alternatives, except the node which updates in the APD model). According to Theorem~\ref{convergence-time}, in the APD model there is a path from each profile to a fixed profile. Combining the above three statements, we can conclude that the SPD model also always reaches a fixed profile. However, the problem of determining the convergence time in the SPD model is posed as an open problem in Section~\ref{conclusion}.

\section{Minimum Cost to Win}
\label{minimum-cost}
\subsection{General Graphs}
\begin{theorem}
\label{min-cost}
In the APD and SPD model on a graph $G=(V,E)$, $\mathcal{MC}(G)\ge \sqrt{n}(\alpha-1)$.
\end{theorem}
\begin{proof}
Consider an arbitrary solution $\mathcal{S}$ and two alternatives $a,b\in\mathcal{A}$. Let $\mathcal{P}$ be a profile where all positions determined by $\mathcal{S}$ are occupied by $a$. Furthermore, for each node if the first position is taken by $a$, assign the second position to $b$; otherwise, assign the first one. Since $\mathcal{S}$ is a solution, in the process starting from $\mathcal{P}$, $a$ must win regardless of the random choices during the process.

Let $L$ be the set of nodes which place $a$ in the first position. If there is a node $v\in L$ such that $|\Gamma_{L}(v)|< |\Gamma_{V\setminus L}(v)|$, then it is possible that in the next round of the APD process, only the order of $a$ and $b$ in $v$ is updated (i.e., $b$ moves to the first position). This is also true for the SPD process since it is possible that all other nodes pick two alternatives which are not adjacent. If we continue this argument repeatedly, we must reach a profile $\mathcal{P}^{\prime}$ where $b$ is ranked first in all nodes except in a non-empty set $L'\subseteq L$, where $a$ is ranked first, and for every node $w\in L'$, $|\Gamma_{L'}(w)|\ge |\Gamma_{V\setminus L'}(w)|$. (Note that $L'$ is non-empty because otherwise the process has reached a profile with $b$ being ranked first in every node, which is a contradiction.) Since $|\Gamma_{L'}(w)|\le |L'|-1$, we have $e\left(L',V\setminus L'\right)\le l'(l'-1)$ for $l':=|L'|$.

There is no set $B\subseteq V\setminus L'$ such that for every $v\in B$, $|\Gamma_{B}(v)|\ge |\Gamma_{V\setminus B}(v)|$ because otherwise all nodes in $B$ will keep $b$ as their first alternative forever, which is a contradiction. Therefore, there is a labeling $u_1,\cdots, u_{n-l'}$ of the nodes in $V\setminus L'$ such that for every $1\le i \le n-l'$, $u_i$ has more neighbors in $L' \cup \{u_j: j<i\}$ than the rest of nodes. If we start from $L'$ and keep adding $u_i$'s to the set one by one, the number of edges on the boundary of the set decreases at least by one after each addition. Note that we do $n-l'$ additions and as we proved in the previous paragraph we start with $e\left(L',V\setminus L'\right)\le l'(l'-1)$ edges. Thus, we can conclude that $n-l' \le l'(l'-1)$, which gives $\sqrt{n}\le l'$. Finally, this implies that the cost of solution $\mathcal{S}$ is at least $l'(\alpha-1)\ge \sqrt{n}(\alpha-1)$ since the first position of nodes in $L'$ are chosen in $\mathcal{S}$.
\end{proof}
\textbf{Tightness.} Consider a clique of size $\sqrt{n}$ and attach $\sqrt{n}-1$ leaves to each node in the clique. This $n$-node graph has a solution of cost $\sqrt{n}(\alpha-1)$, namely the placement which selects the first position in all nodes in the clique and the last position in the others. A proof is given in the appendix, Section~\ref{min-cost-appendix}.

\subsection{Expander Graphs}
\label{expander}

Our results about expander graphs are presented in Theorem~\ref{thm-expander}. The main ingredient of the proof of Theorem~\ref{thm-expander} is Proposition~\ref{proposition-expander}. More precisely, to prove the theorem, we need to repeatedly apply the proposition; this is discussed in the appendix, Section~\ref{thm-expander-appendix}. To prove the proposition, we first need to prove Lemma~\ref{bad}, which uses a variant of the expander mixing lemma presented in Lemma~\ref{mixing}.
\begin{lemma}[Theorem 9.2.4 in~\cite{alon2016probabilistic}]
\label{mixing}
Let $G=(V,E)$ be a $d$-regular graph. For a node set $S\subset V$, we have $\sum_{v\in V}\left(|\Gamma(v)\cap S|-\frac{d|S|}{n}\right)^2\le (\lambda d)^2|S|\left(1-\frac{|S|}{n}\right)$.
\end{lemma}

Let the set of ``bad'' nodes $B_{ab}^{\mathcal{P}}$ be the set of nodes for which at least half of the neighbors rank $b$ higher than $a$ in profile $\mathcal{P}$. For an order $\succ$, let $B_{\succ}^{\mathcal{P}}$ be the union of $B_{ab}^{\mathcal{P}}$'s for every two distinct alternatives $a,b\in \mathcal{A}$ such that $a\succ b$. Furthermore, we define $b_{ab}^{\mathcal{P}}:=|B_{ab}^{\mathcal{P}}|$ and $b_{\succ}^{\mathcal{P}}:=|B_{\succ}^{\mathcal{P}}|$.

\begin{lemma}
\label{bad}
Consider a profile $\mathcal{P}$ on a $d$-regular graph $G=(V,E)$. If the order $\succ$ is $\delta$-Condorcet in $\mathcal{P}$, for some $\delta>0$, then $b_{\succ}^{\mathcal{P}}\le (2\alpha\lambda/\delta)^2\bar{n}_{\succ}^{\mathcal{P}}$.
\end{lemma}
\begin{proof}
Consider two arbitrary alternatives $a,b\in \mathcal{A}$ such that $a \succ b$. By applying Lemma~\ref{mixing} for sets $V_{ba}^{\mathcal{P}}$ and $V_{ab}^{\mathcal{P}}$ and then combining the two obtained inequalities, we get
\begin{equation}
\label{eq-mixing}
\begin{split}
        \mathlarger{\mathlarger{\sum}}_{v\in V}\left(\left(\left|\Gamma(v)\cap V_{ba}^{\mathcal{P}}\right|-\frac{dn_{ba}^{\mathcal{P}}}{n}\right)^2+ \left(\left|\Gamma(v)\cap V_{ab}^{\mathcal{P}}\right|-\frac{dn_{ab}^{\mathcal{P}}}{n}\right)^2\right)\le \\ \left(\lambda d\right)^2\left(n_{ba}^{\mathcal{P}}\left(1-\frac{n_{ba}^{\mathcal{P}}}{n}\right)+n_{ab}^{\mathcal{P}}\left(1-\frac{n_{ab}^{\mathcal{P}}}{n}\right)\right).
\end{split}
\end{equation}
\textbf{Fact 1.} For reals $x>y$ and $z\ge w$, it is straightforward to show that $(z-y)^2+(w-x)^2\ge (x-y)^2/2$.

For $v\in B_{ab}^{\mathcal{P}}$, we have $\left| \Gamma(v)\cap V_{ba}^{\mathcal{P}} \right|\ge \left| \Gamma(v)\cap V_{ab}^{\mathcal{P}} \right|$, by the definition of $B_{ab}^{\mathcal{P}}$. By setting $x=dn_{ab}^{\mathcal{P}}/n$ and $y = dn_{ba}^{\mathcal{P}}/n$ in Fact 1, every node $v\in B_{ab}^{\mathcal{P}}$ contributes at least $d^2(n_{ab}^{\mathcal{P}}-n_{ba}^{\mathcal{P}})^2/(2n^2)\ge d^2\delta^2/2$ (where we used $n_{ab}^{\mathcal{P}}\ge n_{ba}^{\mathcal{P}}+\delta n$) to the left hand-side of Equation~(\ref{eq-mixing}). Thus, the left-hand side is at least $b_{ab}^{\mathcal{P}} d^2\delta^2/2$. Furthermore, using $n_{ab}^{\mathcal{P}}=n-n_{ba}^{\mathcal{P}}$, we can upper-bound the right-hand side of Equation~(\ref{eq-mixing}) by $(\lambda d)^22n_{ab}^{\mathcal{P}}n_{ba}^{\mathcal{P}}/n\le (\lambda d)^22n_{ba}^{\mathcal{P}}$. Hence, we have $b_{ab}^{\mathcal{P}} \le \left(2\lambda/\delta\right)^2n_{ba}^{\mathcal{P}}$. Now, using the fact that $n_{ba}^{\mathcal{P}}\le \bar{n}_{\succ}^{\mathcal{P}}$ and summing up over all ${\alpha \choose 2}\le \alpha^2$ possible choices of the alternatives $a,b$, we conclude that $ b_{\succ}^{\mathcal{P}}\le (2\alpha\lambda/\delta)^2\bar{n}_{\succ}^{\mathcal{P}}$.
\end{proof}

Let a \emph{phase} be a sequence of ${\alpha \choose 2}$ rounds. For our analysis here, we break down the process into phases. Let $\mathcal{P}_{t}$, for $t\ge 1$, denote the profile at the end of the $t$-the phase here (instead of round) and use the shorthand $n_{\succ}^{t}$ for $n_{\succ}^{\mathcal{P}_t}$ and $\bar{n}_{\succ}^{t}$ for $\bar{n}_{\succ}^{\mathcal{P}_t}$.

\begin{proposition}
\label{proposition-expander}
Consider the SPD model on a $d$-regular graph $G=(V,E)$. Assume that $\lambda(G)\le \beta$ for a sufficiently small constant $\beta>0$. If the order $\succ$ is $\epsilon$-Condorcet in $\mathcal{P}^{t-1}$, for some $\epsilon>0$, then $\succ$ is $\epsilon$-Condorcet in $\mathcal{P}^{t}$ and $\bar{n}_{\succ}^{t}\le (1-f(\alpha))\bar{n}_{\succ}^{t-1}$ w.p. at least $1-\exp(-\bar{n}_{\succ}^{t-1}f(\alpha)/8)-\exp(-\Theta(n))$, where $f(\alpha):= 1/(8\alpha^{(\alpha^2)})$.
\end{proposition}
\textsc{Proof Sketch.}
Let $\mathcal{P}_t^{(i)}$ for $i=1,\cdots, r:={\alpha \choose 2}$ be the profile at the end of the $i$-th round in phase $t$ and $\mathcal{P}_t^{(0)}=\mathcal{P}_{t-1}$. Furthermore, let us use the shorthand $\bar{n}_{\succ}^{t,i}$ for $\bar{n}_{\succ}^{\mathcal{P}_t^{(i)}}$ and $b_{\succ}^{t,i}$ for $b_{\succ}^{\mathcal{P}_t^{(i)}}$. We claim that in the profile $\mathcal{P}_t^{(i)}$, the order $\succ$ is $\delta$-Condorcet for $\delta\ge \epsilon -2i*g_{\alpha,\epsilon}$ and $b_{\succ}^{t,i}\le g_{\alpha,\epsilon} \bar{n}_{\succ}^{t,i}$, for $g_{\alpha,\epsilon}:=2\epsilon f(\alpha)/(re^{1/2})$. This statement can be proven by induction over $i$. For the base case of $i=0$, we need to apply Lemma~\ref{bad} and utilize the fact that $\beta$ is a sufficiently small constant (more precisely, we need $\beta\le \sqrt{g_{\alpha,\epsilon}} (\epsilon/4\alpha)$, which is acceptable since both $\epsilon$ and $\alpha$ are constants). For the induction step, we additionally need to use the fact that moving from round $i-1$ to $i$ at most $b_{\succ}^{t,i-1}\le g_{\alpha,\epsilon}n$ nodes could have changed a pair of alternatives in violation with order of $\succ$ and the fact that $\epsilon -2i*g_{\alpha,\epsilon}\ge \epsilon/2$ for $i\le r$.

By the above statement, we have $\sum_{i=0}^{r-1}b_{\succ}^{t,i}\le g_{\alpha,\epsilon}\sum_{i=0}^{r-1}\bar{n}_{\succ}^{t,i}$. Using $\bar{n}_{\succ}^{t,i}\le (1+g_{\alpha,\epsilon})\bar{n}_{\succ}^{t,i-1}$ and $\bar{n}_{\succ}^{t,0}=\bar{n}_{\succ}^{t-1}$ gives us $\sum_{i=0}^{r-1}b_{\succ}^{t,i}\le g_{\alpha,\epsilon}\bar{n}_{\succ}^{t-1}\sum_{i=0}^{r-1}(1+g_{\alpha,\epsilon})^i$. Replacing the values of $r$ and $g_{\alpha,\epsilon}$, using the estimate $1+x\le \exp(x)$, and some small calculations, we get $\sum_{i=0}^{r-1}b_{\succ}^{t,i}\le 2f(\alpha)\bar{n}_{\succ}^{t-1}$.

So far, we showed that there is a set $L:=\{u_1,\cdots, u_{l}\}$ for some $l\ge (1-2f(\alpha))\bar{n}_{\succ}^{t-1}$ such that for every $u\in L$ and every pair of alternatives, more than half of $u$'s neighbors agree with the relative order imposed by $\succ$ during all rounds in phase $t$ and regardless of the random choices. Assume that the Bernoulli random variable $x_k$, for $1 \le k \le l$, is 1 if and only if the pairs of alternatives selected in node $u_k$ during the ${\alpha \choose 2}$ rounds in the phase $t$ match exactly the pairs chosen by the Bubble sort. For example, if $\alpha=4$, then the pairs are chosen are respectively $(4,3)$, $(3,2)$, $(2,1)$, $(4,3)$, $(3,2)$, $(4,3)$, where $(i,j)$ means alternatives in position $i$ and $j$ are selected. It is straightforward to observe that if $x_k=1$, then node $u_k$ will have order $\succ$ at the end of phase $t$ (by a simple induction over the alternatives). Note that the other direction is not true, i.e., $u_k$ having order $\succ$ at the end does not imply $x_k=1$. The variables $x_k$'s are defined in this way to ensure that they are independent. Let us define $X:=\sum_{k=1}^{l}x_k$. Since $\mathbb{P}[x_k=1]=1/{\alpha \choose 2}^{{\alpha \choose 2}}$, we have $\mathbb{E}[X]=l/{\alpha \choose 2}^{{\alpha \choose 2}}\ge (1-2f(\alpha))\bar{n}_{\succ}^{t-1}/\alpha\alpha^{(2\alpha^2)}\ge \bar{n}_{\succ}^{t-1}/(2\alpha^{(2\alpha^2)})$, where in the last inequality we used $f(\alpha)\le 1/4$. Applying the Chernoff bound (see Section~\ref{preliminaries}) gives us $\mathbb{P}[X\le 3\bar{n}_{\succ}^{t-1}/(8\alpha^{(2\alpha^2)})]\le \exp(-\bar{n}_{\succ}^{t-1}/(64\alpha^{(2\alpha^2)}))$. Since the number of nodes with order $\succ$ increases at least by $X$ and decreases at most by $2f(\alpha)\bar{n}_{\succ}^{t-1}$, we can conclude that $\bar{n}_{\succ}^{t}\le (1-f(\alpha))\bar{n}_{\succ}^{t-1}$ w.p. at least $1-\exp(-f(\alpha)\bar{n}_{\succ}^{t-1}/8)$.

To finish the proof, it remains to show that $\succ$ is $\epsilon$-Condorcet in $\mathcal{P}^{t}$. Consider two alternatives $a,b$ such that $a\succ b$. If the difference between the number of nodes which prefer $a$ over $b$ and the nodes which prefer $b$ over $a$ in $\mathcal{P}_{t-1}$ is at least $\epsilon n+ 2f(\alpha)\bar{n}_{\succ}^{t-1}$, then the profile remains $\epsilon$-Condorcet after phase $t$ since at most $2f(\alpha)\bar{n}_{\succ}^{t-1}$ nodes can move $b$ over $a$ during phase $t$. Otherwise, there are linearly many nodes which rank $b$ over $a$ but their neighborhood ranks $a$ higher during all rounds in phase $t$. Similar to above, using Chernoff bound we can prove w.p. $1-\exp(-\Theta(n))$, the number of nodes which move $a$ higher than $b$ is at least $2f(\alpha)\bar{n}_{\succ}^{t-1}$, which guarantees the $\epsilon$-Condorcet property w.r.t. $a$ and $b$. By a union bound, this holds for all pairs of alternatives w.p. at least $1-{\alpha \choose 2} \exp(-\Theta(n))=1-\exp(-\Theta(n))$. \qed

\begin{theorem}
\label{thm-expander}
Consider the SPD model on a $d$-regular graph $G=(V,E)$. Assume that $\lambda(G)\le \beta$ for a sufficiently small constant $\beta>0$. If the order $\succ$ is $\epsilon$-Condorcet in the initial profile $\mathcal{P}_0$, for some constant $\epsilon>0$, then $\succ$ wins in $\mathcal{O}(\log n)$ rounds w.h.p.
\end{theorem}
\textbf{Minimum Cost Solution.} One can infer that for a graph $G$ which satisfies the conditions of Theorem~\ref{thm-expander}, $\mathcal{MC}(G)\ge (1/2-\epsilon)n(\alpha -1)$, for any constant $\epsilon>0$, because a placement with cost at most $(1/2-\epsilon)n(\alpha -1)$ does not occupy the first position in at least $(1/2+\epsilon)n$ nodes. (Some details are left out.)

\textbf{Tightness of Convergence Time.} Consider a $d$-regular graph $G$, with constant $d$, which satisfies the conditions of Theorem~\ref{thm-expander}. Consider an arbitrary node $v$ and let $B$ be the set of nodes whose distance from $v$ is at most $t^*=(\log_d n)/2-1$. Note that $|B|\le \sum_{i=0}^{t^*}d^{i}\le d^{t^*+1}=\sqrt{n}$. Assume that initially all nodes in $B$ have order $\succ^{\prime}$ and nodes in $V\setminus B$ have order $\succ\ne \succ^{\prime}$. This clearly satisfies the Condorcet condition of Theorem~\ref{thm-expander}, and thus $\succ$ wins in $\mathcal{O}(\log n)$ rounds w.h.p. However, it must take at leas $t^*=\Omega(\log n)$ rounds until $v$ has order $\succ$. Therefore, the logarithmic bound is asymptotically tight.

\textbf{Conditions.} Can some of the conditions of Theorem~\ref{thm-expander} be lifted or relaxed? The statement requires $\alpha$ to be a constant, $\lambda$ to be smaller than a sufficiently small constant $\beta$ (which enforces $d\ge C$ for an arbitrarily large constant $C$), the graph to be regular, and the process to be synchronous. In the appendix, Section~\ref{limit-appendix}, we argue how most of these can be potentially removed or relaxed using our proof techniques in combination with other tools.

\subsection{Number of Solutions}
We study the number of solutions $\mathcal{NS}(G)$ for a graph $G$. We particularly prove a tight bound on $\mathcal{NS}(C_n)$ for a cycle $C_n$ in Theorem~\ref{ns}, whose full proof is given in the appendix, Section~\ref{ns-appendix}.
\begin{theorem}
\label{ns}
In both the APD and SPD model on a cycle $C_n$, $\mathcal{MC}(C_n)=(2\lfloor n/3\rfloor + (n\mod 3)) (\alpha -1)$ and $\mathcal{NS}(C_n)=\Tilde{\Theta}(\psi^n)$ for some $\psi\in (\alpha^{1/3}, \alpha^{1/3}+0.22)$, where $\Tilde{\Theta}$ hides polynomial terms in $n$ and $\alpha$.
\end{theorem}
\textsc{Proof Sketch.} 
Let $\mathcal{S}_n$ be the set of all placements on $C_n$ such that for every three consecutive nodes for at least two of them the first position is chosen. We can show that $\mathcal{S}_n$ is equal to the set of all solutions in both the APD and SPD model on $C_n$. Thus, $\mathcal{MC}(C_n)$ is equal to the minimum cost among the elements of $\mathcal{S}_n$, which we claim to be $(2\lfloor n/3\rfloor+(n\mod 3))(\alpha-1)$.

Let $s(n):=|\mathcal{S}_n|$ and $p(n)$ be $s(n)$ but for a path $P_n$ instead of $C_n$. We observe that $p(n-2)\le s(n)\le p(n)$. Thus, to approximate $s(n)$, it suffices to calculate $p(n)$. For that, we need to solve the homogeneous linear difference equation $p(n)=p(n-1)+(\alpha-1)p(n-3)$. By solving the characteristic equation $\lambda^3-\lambda^2-(\alpha-1)=0$ (derived from $p(n)=p(n-1)+(\alpha-1)p(n-3)$) and some standard calculations, we get $p(n)=\Tilde{\Theta}(\psi^n)$ for some $\psi\in (\alpha^{1/3}, \alpha^{1/3}+0.22)$. Finally, since $p(n-2)\le s(n)\le p(n)$ and $\Tilde{\Theta}$ hides polynomial terms in $\alpha$, we can conclude that $s(n)=\Tilde{\Theta}(\psi^n)$. Please refer to the appendix, Section~\ref{ns-appendix}, for a complete proof. \qed

For an empty graph $E_n$ (a graph with $n$ nodes and zero edges), we have $\mathcal{NS}(E_n)=1$ because the only possible solution is to choose the first position in all nodes. For a complete graph $K_n$, we have $\mathcal{NS}(K_n)=\sum_{k=\lfloor n/2\rfloor+1}^{n} {n \choose k}(\alpha-1)^{n-k}$ because a solution requires at least $\lfloor n/2\rfloor+1$ nodes to choose the first position. As an upper bound, we have $\mathcal{NS}(K_n)\le (\alpha-1)^{n/2}\sum_{k=0}^n{n \choose k}=(\alpha-1)^{n/2}2^n$. Furthermore, $\mathcal{NS}(K_n)\ge (\alpha-1)^{\lceil n/2\rceil-1}{n \choose \lfloor n/2\rfloor+1}=(\alpha-1)^{\lceil n/2\rceil-1}*\Theta(2^n/\sqrt{n})$, where we used the estimate ${n \choose \lfloor n/2\rfloor+1}=\Theta(2^n/\sqrt{n})$. Therefore, we have $\mathcal{NS}(K_n)=\Tilde{\Theta}((4(\alpha-1))^{n/2})$.

It would be interesting to characterize the graph parameters which control the value of $\mathcal{NS}(G)$. Since we have proven that $\mathcal{NS}(E_n)\le \mathcal{NS}(C_n)\le \mathcal{NS}(K_n)$, one might be tempted to conjecture that $\mathcal{NS}(G)$ grows in the number of edges in $G$. However, this is not true; for example, one can check that for a star graph $S_n$ (which has only $n-1$ edges) there are almost as many solutions as in $K_n$. This is true because a placement is a solution if the internal node and at least $\lceil (n-1)/2\rceil$ of the $n-1$ leaves choose the first position. For $\mathcal{NS}(G)$ to be large, many placements should result in a solution. Roughly speaking, this implies that many nodes should have a similar level of ``power'', which can be captured in terms of graph parameters like regularity and vertex-transitivity. Investigating this problem further is left to future work.

\section{Random Preference Diffusion}
\label{rc}
\begin{theorem}
\label{rc-thm}
Consider the Random PD on a graph $G=(V,E)$.
\begin{enumerate}[I]
    \item The process always reaches a fixed profile, where the order of all nodes is equal to some order $\succ^f$.
    \item For $a,b\in \mathcal{A}$, $\mathbb{P}[a\succ^f b]=Z_{0}^{ab}/(2m)$, where $Z_{0}^{ab}$ is the summation of the degree of nodes which rank $a$ higher than $b$ initially.
    \item The process takes $\mathcal{O}(n^4)$ rounds in expectation.
\end{enumerate}
\end{theorem}

The proof of part III is given in the appendix, Section~\ref{rc-appendix}. It uses the results from~\cite{coppersmith1993collisions} regarding the meeting time in some random walks. We also discuss that a stronger bound of $\mathcal{O}(n\log n)$ exists for graphs with strong expansion properties. 

\textsc{Proof of I and II.}
We prove that from any profile, there is a non-zero probability to reach a fixed profile. Therefore, the process must reach a fixed profile eventually. Consider an arbitrary profile $\mathcal{P}$. Let $v$ be an arbitrary node and let $\succ$ be its order in $\mathcal{P}$. Define $\Gamma_t(v)$ to be the nodes whose distance from $v$ is equal to $t$. There is a non-zero probability that in the next round, all nodes in $\Gamma_1(v)$ (i.e., the neighbors of $v$) choose order $\succ$ and all other nodes keep their order unchanged. Repeating the same argument for $\Gamma_t(v)$ and $t\ge 2$, we can conclude that starting from $\mathcal{P}$, it is possible that the process reaches the profile where all nodes hold the order $\succ$.


Let the random variable $Z_t^{\succ}$, for $t \in\mathbb{N}$ and some order $\succ$, denote the summation of the degree of nodes whose order is equal to $\succ$ in $\mathcal{P}_t$. We prove that the sequence $Z_0^{\succ}, Z_1^{\succ}, Z_2^{\succ},\cdots$ is a discrete-time martingale, i.e., $\mathbb{E}[Z_t^{\succ}|Z_0^{\succ}, Z_1^{\succ}, \cdots, Z_{t-1}^{\succ}]=Z_{t-1}^{\succ}$. Since a node $v$ keeps its order w.p. $(1-q)$ and selects $\succ$ from its neighbors w.p. $q|\Gamma(v)\cap V_{\succ}^{\mathcal{P}_{t-1}}|/|\Gamma(v)|$, then $\mathbb{E}\left[Z_t^{\succ}|Z_0^{\succ}, Z_1^{\succ}, \cdots, Z_{t-1}^{\succ}\right]$ is equal to
\begin{align*}
&\sum_{v\in V_{\succ}^{\mathcal{P}_{t-1}}} (1-q) d(v) +\sum_{v\in V} q d(v) \frac{|\Gamma(v)\cap V_{\succ}^{\mathcal{P}_{t-1}}|}{d(v)}=
\\
&(1-q)\sum_{v\in V_{\succ}^{\mathcal{P}_{t-1}}} d(v) + q \sum_{v\in V} |\Gamma(v)\cap V_{\succ}^{\mathcal{P}_{t-1}}| =
\sum_{v\in V_{\succ}^{\mathcal{P}_{t-1}}} d(v) = Z^{\succ}_{t-1}
\end{align*}
where we used $\sum_{v\in V} |\Gamma(v)\cap V_{\succ}^{\mathcal{P}_{t-1}}|=\sum_{v\in V_{\succ}^{\mathcal{P}_{t-1}}} d(v)$.

Let $Z_f^{\succ}$ be the summation of the degree of nodes whose order is equal to $\succ$ in the final profile. We proved that the expected value of $Z_f^{\succ}$ is equal to $Z_0^{\succ}$. (Note that $Z_0^{\succ}$ is a fixed value.) Putting this in parallel with the fact that the process always reaches a fixed profile, we can conclude that the probability that $\succ$ wins is equal to $Z^{\succ}_0/2m$ since the summation of all degrees is equal to $2m$.

Consider two alternatives $a,b\in \mathcal{A}$. Let $\pmb{O}_{ab}$ be the set of orders which rank $a$ higher than $b$. Since the event that an order $\succ$ wins is disjoint from the event that $\succ^{\prime}\ne \succ$ wins, we have
\begin{equation*}
    \mathbb{P}[a\succ^{f} b]= \sum_{\succ \in \pmb{O}_{ab}}\mathbb{P}[\succ^f\textrm{is equal to}\succ]=\sum_{\succ \in \pmb{O}_{ab}}\frac{Z_0^{\succ}}{2m}=\frac{Z_{0}^{ab}}{2m}. \qed
\end{equation*}

\section{Experiments}
\label{experiments}
\begin{figure*}
\includegraphics[scale=0.55]{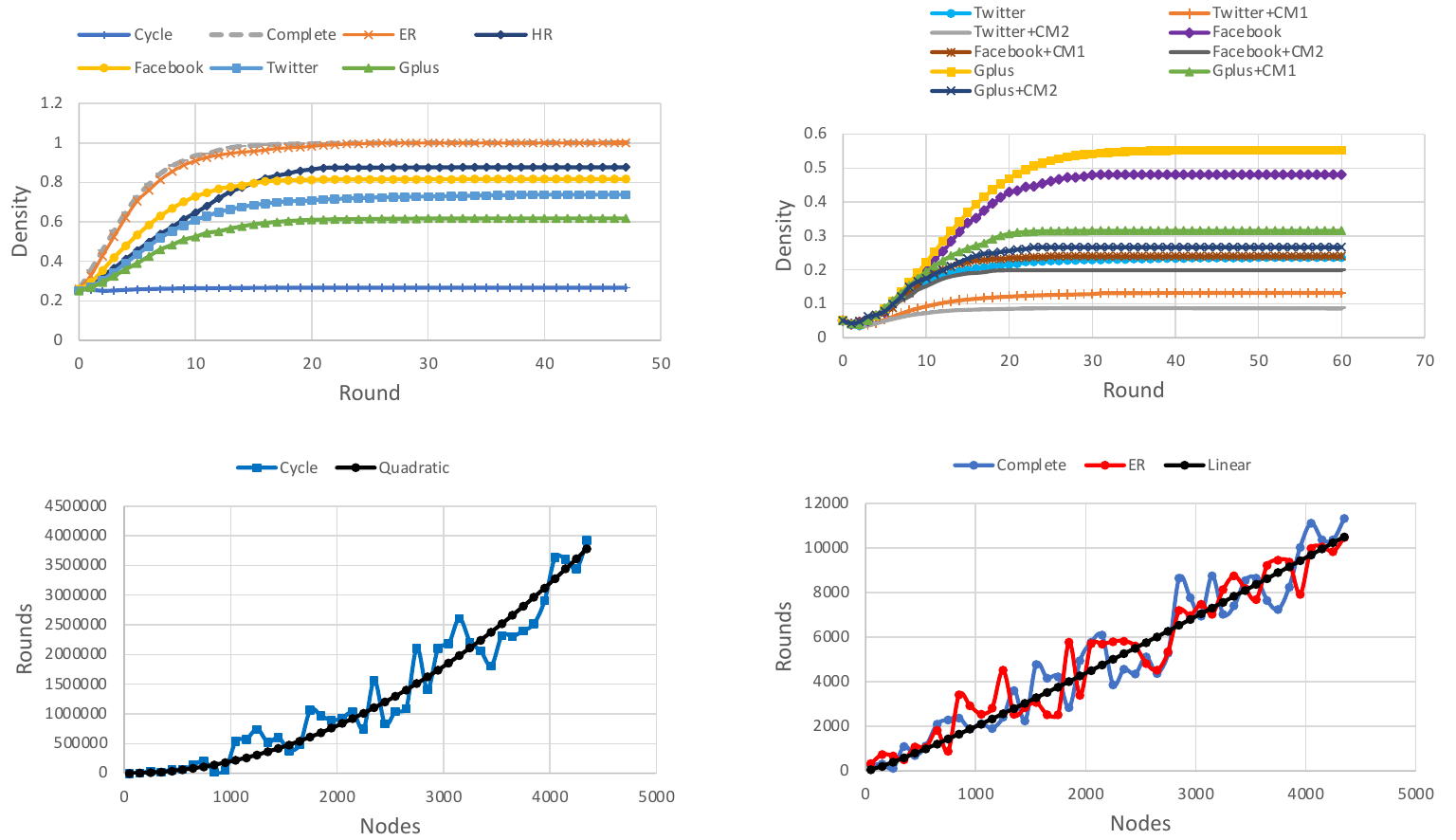}
  \caption{The density of order $\succ$ in the SPD process: (top-left) with an initial random profile with a bias towards $\succ$, (top-right) with only $5\%$ highest degree nodes choosing $\succ$, where $CM_1$ is the countermeasure to add $10$ random edges to each node and $CM_2$ is to give more weight to the preference of more similar neighbors. The convergence time of the Random PD process starting with a u.a.r. profile: (down-left) on cycle graph, (down-right) on complete and ER random graph.}
  \label{experiments-fig}
\end{figure*}

We have developed and conducted a set of experiments for the PD process on different synthetic graph structures (e.g., cycle and ER random graph) and graph data from real-world social networks (e.g., Facebook and Twitter). While some of our experiments focus on the aspects that we already studied theoretically, such as convergence time and (randomized) minimum cost solutions, we leverage our experiments to investigate other aspects of the model, such as introducing countermeasures to make the process more ``democratic'' and to create foundations for proposing some future research directions. (It is worth to stress that our theoretical findings are self-contained.)

\textbf{Set-up.} For the real-world graph data, we rely on the publicly available network datasets from~\cite{leskovec2014snap}. In the Facebook graph we have $n=4039$ and $m=88234$, in the Twitter graph $n=81306$ and $m=1768149$, and in the Gplus graph $n=107614$ and $m=13673453$. For all of our synthetic graph structures, including cycle and complete graph, we set $n=4039$, which is equivalent to the number of nodes in the Facebook graph. For the ER random graph $\mathcal{G}_{n,q}$, we set $q$ such that the expected degree of a node is equal to the average degree in the Facebook graph. Our experiments have also been conducted on Hyperbolic Random (HR) graph, a graph model that simulates real-world social networks. (HR graph is known to have fundamental properties observed in real-world networks such as small diameter, clustering property, and power law degree distribution. See~\cite{krioukov2010hyperbolic} for more details.)
The parameters are set such that the number of nodes and average degree are comparable to the ones in the Facebook graph. We also need to provide the exponent of the power-law degree distribution $\beta$ and the temperature $T$ as the input parameters. We set $\beta= 2.5$ and $T = 0.6$.

All of our experiments were executed 20 times and then the average output was considered. Furthermore, all experiments were conducted on 64-bit Ubuntu 18.04 LTS with an Intel Core i7-3930K 3.2 GHz CPU and 64 GB RAM and using Go programming language.

\textbf{Condorcet Winner.} Consider the SPD model with three alternatives $a, b, c$. We start with a random profile $\mathcal{P}_0$ where each node picks one of the 6 possible orders u.a.r., except that a $10\%$ advantage is given to the order $\succ$ which puts $a$, $b$, and $c$ in positions $1$, $2$, and $3$, respectively. For large values of $n$, the order $\succ$ almost surely will be $\epsilon$-Condorcet in $\mathcal{P}_0$, for some $\epsilon>0$. As demonstrated in Figure~\ref{experiments-fig} (top-left), for the complete graph and ER random graph, which enjoy strong expansion properties, all nodes adopt the order $\succ$ (where the convergence time is, arguably, of logarithmic order since $\log 4039 =8.3$). This is consistent with the statement of Theorem~\ref{thm-expander}. For the three studied social graphs, a somewhat similar behavior is observed. This can be explained by the fact that while the real-world social networks are not usually perfect expanders, like an ER random graph, they definitely enjoy a certain level of expansion. The HR graph (with parameters comparable to the Facebook graph) demonstrates a similar behavior (which provides support that the HR graph is suitable for modeling real-world social networks). On the other hand, in a cycle graph, which suffers from weak expansion properties, the density of nodes with order $\succ$ remains almost unchanged. We should mention that this behavior in cycles is not solely caused by weak expansion properties, but rather its combination with regularity and vertex-transitivity.

\textbf{Minimum Cost Solution.} Our explanations on Figure~\ref{experiments-fig} (top-left) imply that if an adversary chooses the set of nodes to bribe at random, engineering the outcome of the process on a real-world social graph would be very costly. However, an adversary might employ more efficient mechanisms. As discussed in Section~\ref{related}, the problem of finding a minimum cost solution for a given graph $G$ has been studied for various majority based models and in most scenarios, it is known to be NP-hard, cf.~\cite{mishra2002hardness,faliszewski2022opinion, elkind2009swap}. However, there exist greedy approaches, usually built on submodular function maximization techniques, which provide polynomial time approximation algorithms. We assume that the adversary bribes the nodes with the highest degrees (this is in some sense similar to what the greedy algorithms do and is computationally very light; thus, it is commonly used in the experimental set-ups, cf.~\cite{avin2019majority, out2021majority}. Our experiments, not included here, yield similar results when we use other measures such as betweenness or closeness instead of degree.). More precisely, we sort the nodes by their degree and the top $5\%$ are assigned the order $\succ$ and the remaining nodes have the reverse order $\succ^{\prime}$. Figure~\ref{experiments-fig} (top-right) demonstrates that for Twitter the density of nodes with order $\succ$ increases from $5\%$ to over $20\%$ and for Facebook and Gplus to almost $50\%$. Note that if the remaining $95\%$ nodes choose their order more uniformly among the $\alpha!$ orders ($6$ in our case) instead of choosing the reverse order $\succ^{\prime}$, then the density of $\succ$ will grow even more aggressively.

\textbf{Countermeasures.} So far, we observed that an adversary which pursues a rather smart strategy, such as bribing the highest degree nodes, can enforce their desired order (or alternative(s)) to a large part of the network with a relatively small cost. A natural question arises is whether one can design effective countermeasures to defeat such an adversary. Since we know that graphs with strong expansion properties are quite resilient against such adversarial attacks, a promising approach is to add some random edges to the graph, with the purpose of making it more expansive. (A similar countermeasure was already introduced in~\cite{out2021majority} for the Majority model.) A second countermeasure that we introduce is that each node gives more weight to the preferences of the nodes which are more ``similar'' to it. For two nodes $v,u$, let $S(v,u):=|\Gamma(v)\cap\Gamma(u)|/(d(v)+d(u))$ be the \emph{similarity} between $v$ and $u$. (This is analogous to similarity coefficients such as Jaccard and Szymkiewicz–Simpson coefficients, cf.~\cite{vijaymeena2016survey}.) Assume that we modify the PD model such that a node $v$ changes the order of two chosen adjacent alternatives $a$ and $b$ if the summation of the similarity of neighbors which disagree with its current order is more than the sum for the ones which agree. 

Figure~\ref{experiments-fig} (top-right) shows that both countermeasures CM1 (adding $10$ random edges to each node) and CM2 (integrating the similarity coefficient into the updating rule) substantially reduce the extent to which the adversary's desired order spreads. In all the three studied social networks, the countermeasure CM2 slightly outperforms CM1.

An ideal countermeasure should have the following three properties. Firstly, it must not demand significant changes in the graph structure or the updating rule. Secondly, it should be easy for the nodes to implement, e.g., it does not require the nodes to have a full knowledge of the graph structure or memorize the history of the process. Last but not the least, such an alternation should be implemented by the nodes instead of being enforced by a central entity, such as a government or an online social platform management team. The last criterion is crucial since it ensures that the countermeasures can be deployed through educating and informing the network members and leveraging their collective decision-making power rather than abusing the excessive power of a central entity which might need to violate fundamental human rights such as freedom of expression to impose the countermeasure.

Our suggested countermeasures clearly satisfy the last two properties, since for their implementation, a node $v$ needs to simply make some random friends in the network or give more weight to the preferences of the ``closer'' friends (i.e., nodes which share a bigger fraction of their neighbors with $v$). Whether the first property is fulfilled or not is up to interpretation. Note that the average degree in our studied social networks is between $43$ and $254$. Thus, requiring all nodes to make 10 random friends is somewhat equivalent to asking them to choose around $4-18\%$ of their neighbors at random. We believe that the second countermeasure is less intrusive and more practical since it only requires the nodes to give more weight to the preferences of the closer friends.

\textbf{Convergence Time in Random PD.} Figure~\ref{experiments-fig} (down-left) visualizes the convergence time of the Random PD model with $q=1/2$ on a cycle $C_n$ for different values of $n$, where in the initial profile each node picks an order, with three alternatives $a, b, c$, u.a.r. Comparing its growth with the function $n^2/5$ (i.e., \emph{Quadratic} line in the figure) suggests that it grows quadratically in $n$.
Figure~\ref{experiments-fig} (down-right) indicates that the convergence time on a complete graph and ER random graph matches the growth of the function $(n\log n)/5$ (i.e., \emph{Linear} line in the figure). This is consistent with the $\mathcal{O}(n\log n)$ bound mentioned in Section~\ref{rc} for graphs with strong expansion properties.

For the Facebook graph (with $n=4039$ ), the Twitter graph (with $n=81306$), and Gplus (with $n=107614$), our experiments output the convergence time of $1941$, $18491$, and $27798$, respectively. This is closer to the linear behavior of the strong expanders rather than the quadratic growth on cycles.

\section{Future Work}
\label{conclusion}

\textbf{Number of Solutions.} We proved tight bounds on the number of solutions $\mathcal{NS}(G)$ for different graphs, e.g. cycles and complete graphs. It would be interesting to determine the graph parameters which govern $\mathcal{NS}(G)$. We argued that regularity and vertex-transitivity are some potential candidates.

\textbf{Countermeasures.} We confirmed the effectiveness of our two proposed countermeasures to neutralize potential adversarial attacks (namely making random friends and giving more weight to the preferences of closer friends) by experimenting on graph data from Twitter, Facebook, and Gplus. It would be valuable to develop powerful countermeasures which satisfy the three properties listed in Section~\ref{experiments}.


\textbf{Convergence Time.} It is left as an open problem to determine the convergence time for the SPD model.

\textbf{Convergence In Directed Graphs.} We proved that the PD process on an undirected graph always reaches a fixed profile. Brill et al.~\cite{brill2016pairwise} proved this for special classes of directed graphs, such as acyclic graphs, but the problem was left open in the general case.


\newpage
\bibliographystyle{ACM-Reference-Format}
\bibliography{ref}

\newpage
\appendix

\section{Eliminated Proofs}
\subsection{Proof of Lemma~\ref{coin-flip}}
\label{coin-flip-appendix}
Let the random variable $x_i$, for $1\le i\le K$, denote the number of coin flips done from when we have seen the $(i-1)$-th head until the $i$-th head. Thus, we have $X=\sum_{i=1}^{K}x_i$. Using the linearity of expectation and $\mathbb{E}[x_i]=1/p$, we get $\mathbb{E}[X] = \mathbb{E}\left[\sum_{i=1}^K x_i\right]=\sum_{i=1}^{K}\mathbb{E}[x_i]=K/p$.

Since $x_i$'s are geometrically distributed, we have $Var[x_i]=(1-p)/p^2$. Thus, applying the fact that $x_i$'s are independent, we have
\begin{equation*}
    Var[X]=Var\left[\sum_{i=1}^{K} x_i\right]= \sum_{i=1}^{K} Var[x_i]= \frac{K(1-p)}{p^2}\le \frac{K}{p^2}.
\end{equation*}
Finally, setting $f=\mathbb{E}[X]/2=K/(2p)$ in the Chebyshev's inequality (see Section~\ref{preliminaries}) we conclude that
\begin{equation*}
    \mathbb{P}\left[\frac{K}{2p}< X<\frac{3K}{2p}\right]\ge 1- \frac{4Var[X]}{\mathbb{E}[X]^2}\ge 1-\frac{4Kp^2}{K^2p^2}=1-\frac{1}{4K}. 
\end{equation*}
\subsection{Tightness of Theorem~\ref{convergence-time}}
\label{convergence-time-tightness}
Let us construct an $n$-node graph $G$. Define $\kappa := \lfloor (n-3)/2 \rfloor$. Consider two cycles $C_w = w_1,\cdots, w_{\kappa}$ and $C_g = g_1,\cdots, g_{\kappa}$. Now add the edge set $\{\{w_i, g_i\}: 1\le i \le \kappa\}$. Furthermore, add a new node $w_{\kappa}^{\prime}$ and connect it to $w_{\kappa}$ as a leaf. Finally, create a clique of size $n-2\kappa-1$ (this is equal to 2 when $n$ is odd and 3 otherwise) and add an edge between $w_1$ and every node in this clique. (See Figure~\ref{figure} for an example.) We observe that this graph has $n$ nodes and $m \le 3\kappa+7=\Theta(n)$ edges (more precisely, $3\kappa+4$ when $n$ is odd and $3\kappa+7$ otherwise). According to Theorem~\ref{convergence-time}, the convergence time of the APD model on $G$ is upper-bounded by $\mathcal{O}(n^2\alpha^4)$ w.h.p. We show that the converge time of the APD model on $G$ can be as large as $(1/2)(\kappa+1) n{\alpha \choose 2}=\Omega(n^2\alpha^2)$.

\begin{figure}[h]
  \centering
  \includegraphics[width=0.5\linewidth]{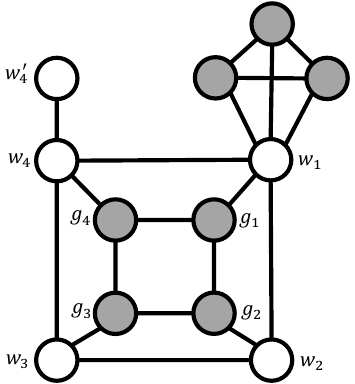}
  \caption{An example of the construction given to prove the tightness of convergence time in Theorem~\ref{convergence-time} (for $n=12$).}
  \label{figure}
  \Description{The construction given to prove the tightness of convergence time in the APD model for $n=12$.}
\end{figure}

Consider two alternatives $a,b\in \mathcal{A}$. Let $\succ_w$ and $\succ_g$ be two orders which are identical except that $\succ_w$ has $a$ in the first position and $b$ in the second while $\succ_g$ has $b$ in the first position and $a$ in the second. We say a node is colored white (resp. gray) if it has the order $\succ_w$ (resp. $\succ_g$). Consider the initial profile where all nodes in $C_w$ are colored white, all nodes in $C_g$ are gray, $w^{\prime}_{\kappa}$ is white, and the nodes in the clique are gray. (See Figure~\ref{figure} for an example.)

All nodes in $C_g$ and the clique remain gray (i.e., keep order $\succ_g$) forever since each of them has at least as many gray neighbors as white. In the initial profile, only node $w_1$ can update. It can update the order of $a$ and $b$ (i.e., change its color from white to gray). Once this happens, the only possible update is for $w_2$ to change from white to gray. With an inductive argument, the process eventually reaches a fully gray profile since nodes $w_1,\cdots,w_{\kappa}, w_{\kappa}^{\prime}$ become gray one after another. Before the process reaches the fully gray profile, in each round there is exactly one update possible which is chosen w.p. $1/(n{\alpha \choose 2})$ and $K=\kappa+1$ such updates need to be executed. Thus, by applying Lemma~\ref{coin-flip}, we can conclude that w.p. at least $1-1/(4(\kappa+1))=1-o(1)$, the process needs at least $(1/2)(\kappa+1) n {\alpha \choose 2}$ rounds to end.

Note that this leaves an $\alpha^2$ gap between the lower and upper bound. In the first glance, it might seem that if we set $\succ_g$ to be the reverse of $\succ_w$, then we get a stronger lower bound, in $\alpha$, but this would improve the lower bound only by a constant factor. Furthermore, the given construction has linearly many edges. It would be interesting to find such constructions for all values of $m$.

\subsection{Tightness of Theorem~\ref{min-cost}}
\label{min-cost-appendix}
We want to build an $n$-node graph $G$ which has a solution of cost $\sqrt{n}(\alpha-1)$. Consider a clique of size $\sqrt{n}$ and attach $\sqrt{n}-1$ leaves to each node in the clique. This graph has $\sqrt{n}+\sqrt{n}(\sqrt{n}-1)=n$ nodes. Assume that initially all nodes in the clique have the alternative $a$ as their first preference and all leaves have $a$ at the bottom of their order. Nodes in the clique will keep $a$ on top forever since they have as many neighbors in the clique (agreeing with them) than outside the clique. For the leaves, every time the alternative $a$ and the alternative on top of it are chosen, $a$ moves up, in both the APD and SPD model. Thus, eventually all nodes will have $a$ as their top choice. This implies that graph $G$ has a solution of cost $\sqrt{n}(\alpha-1)$ in both the APD and SPD model.

\subsection{Proof Sketch of Theorem~\ref{thm-expander}}
\label{thm-expander-appendix}
Let $t^*$ satisfy the equality $(1-f(\alpha))^{t^*}\bar{n}_{\succ}^{0}= (\log\log n)^2$. Then, $t^*= \log _{1/(1-f(\alpha))}(\bar{n}_{\succ}^{0}/(\log \log n)^2)=\mathcal{O}(\log n)$, where in the last equality we used $\bar{n}_{\succ}^{0}\le n$ and $1/(1-f(\alpha))$ is a constant larger than 1. Applying Proposition~\ref{proposition-expander} repeatedly implies that after $t$ phases for some $t\le t^*$, the process reaches a profile where at most $(\log \log n)^2$ nodes do not have the order $\succ$ w.p. at least $1-\sum_{i=1}^{t^*}\exp(-f(\alpha) (1-f(\alpha))^{i-1}\bar{n}_{\succ}^{0}/8)-\sum_{i=1}^{t^*} \exp(-\Theta(n))$. This probability is at least $1- \mathcal{O}(\log n) \exp(-f(\alpha)(\log\log n)^2/8)-\mathcal{O}(\log n)\exp(-\Theta(n))=1-o(1)$. Note that since each phase is a constant number of rounds, the number of rounds required is also in $\mathcal{O}(\log n)$. (To be fully precise, we would have needed to define an event $A_i$ for each application of Proposition~\ref{proposition-expander} and condition on events $A_1, \cdots, A_{i-1}$, but the calculations would still be the same.)

It remains to argue that from a profile where only $(\log\log n)^2$ nodes do not hold the order $\succ$, the process reaches a profile with at most $d/4$ such nodes in $\mathcal{O}(\log\log\log n)$ rounds w.h.p. (A few rounds after that, all nodes have $\succ$ since for every node at least a $3/4$ fraction of its neighbors hold $\succ$.) We do not provide the full proof here due to the space constraints, but sketch the main ideas. We first need to use Proposition~\ref{proposition-expander} and argue that the error probability is at most $1/4$. This can be done by using $\bar{n}_{\succ}^{t-1}\ge d/4$ and the fact that we can set $d\ge C$ for an arbitrarily large constant $C$ (since we assume $\lambda\le \beta$ and $\lambda\ge 1/\sqrt{d}$ (cf.~\cite{alon2016probabilistic}), choosing $\beta$ sufficiently small gives us a sufficiently large constant $C$). Thus, w.p. $3/4$ the number of nodes which do not have $\succ$ decreases by a constant factor $1-f(\alpha)$, otherwise it increases by a constant factor in the worst case. This can be modeled as a gambler who has $(\log\log n)^2$ dollars and in each round they lose a constant fraction of their \textit{current} money w.p. $3/4$ and gain a constant fraction w.p. $1/4$. Using Chernoff bound (see Section~\ref{preliminaries}), one can show that they will become bankrupt in $\mathcal{O}(\log\log\log n)$ rounds w.h.p. (To use Chernoff bound, we need to define the gambler's process in a way that it upper-bounds our original process, but it gives us the required independence. Another option, of course, is to use inequalities which allow some level of dependence, such as the Azuma-Hoeffding inequality~\cite{dubhashi2009concentration}, and work with the original process directly.)

\subsection{Conditions of Theorem~\ref{thm-expander}}
\label{limit-appendix}

\textbf{Bounds on $\alpha$, $\lambda$, and $d$.} In Theorem~\ref{thm-expander}, we assume that $\alpha$ is constant. Our proof can be adapted to cover the case of $\alpha$ growing ``slowly'' in $n$. From a practical perspective, this assumption is not too unreasonable since in reality the number of alternatives is much smaller than the number of individuals in a network. We also need the assumption that $\lambda\le \beta$, but this is not very restrictive either since we know there are graphs where $\lambda =\mathcal{O}(1/\sqrt{d})$, e.g. random $d$-regular graphs~\cite{alon2016probabilistic}. Furthermore, as mentioned in the proof, requiring $\lambda \le \beta$ for a sufficiently small constant $\beta$ imposes a constant lower bound on $d$ since $1/\sqrt{d}\le \lambda$ (see~\cite{alon2016probabilistic}).

\textbf{Irregular Graphs.} To avoid extra complications in the proof, we focused on regular graphs. However, we can utilize a more general variant of Lemma~\ref{mixing} (cf.~\cite{alon2016probabilistic}) to provide a similar statement for irregular graphs. The main difference is that the statement, especially the definition of Condorcet, needs to be modified to use the summation of the degree of nodes which hold an order $\succ$ rather than the number of such nodes.

\textbf{Asynchronous Model.} Theorem~\ref{thm-expander} covers only the synchronous set-up. To prove the analogous result for the asynchronous PD model, one possible approach is to break the process into \emph{super phases}, where a super phase is a sequence of $p(n)$ phases, for a suitable choice of $p(n)=\Theta(n)$, and then try to use the proof from Theorem~\ref{thm-expander} for super phases instead of phases. Another plausible approach is to use the techniques from the bounded biased random walk on the integers and the Azuma-Hoeffding inequality, see~\cite{dubhashi2009concentration}. However, we do not provide this proof here.

\subsection{Proof of Theorem~\ref{ns}}
\label{ns-appendix}

Let $\mathcal{S}_n$ be the set of all placements on a cycle $C_n = v_1,\cdots, v_n$ such that for every three consecutive nodes for at least two of them the first position is chosen. We claim that $\mathcal{S}_n$ is equal to the set of all solutions in both the APD and SPD model on $C_n$. Consider an arbitrary solution $S$. Build the profile $\mathcal{P}$ which puts the alternative $a$ in positions determined by $S$ and then put $b$ in the highest position possible. If $S$ is not in $\mathcal{S}_n$, there are three consecutive nodes $w_1, w_2, w_3$ so that at least two of them rank $b$ in the first position. If these two nodes are adjacent, say $w_1, w_2$, they keep $b$ as their first alternative forever, in both the APD and SPD model starting from $\mathcal{P}$, which is a contradiction. If the two nodes are not adjacent, namely $w_1, w_3$, then there is a non-zero probability that $w_1$ and $w_3$ do not update while $w_2$ compares $a$ and $b$ and moves $b$ on top. (Note that this is possible in the SPD model too since $w_1$ and $w_3$ might pick pairs which are not adjacent on their order.) Again, $w_1$ and $w_2$ will keep $b$ on top forever, which is a contradiction. Thus, $S$ must be in $\mathcal{S}_n$.

Now, consider a placement $S'$ in $\mathcal{S}_n$. We show that $S'$ is a solution. Consider a profile $\mathcal{P}$ where the alternative $a$ is placed according to $S'$. Every node which has $a$ in the first position has at least one neighbor which does the same. Therefore, it keeps $a$ on top forever, starting from $\mathcal{P}$ in both the APD and SPD model. All other nodes have two neighbors which rank $a$ first, thus $a$ keeps leveling up in those nodes until all nodes will have $a$ in the first position.

So far, we proved $\mathcal{S}_n$ is equal to the set of all solutions. Thus, $\mathcal{MC}(C_n)$ is equal to the minimum cost among the elements of $\mathcal{S}_n$. We claim that $\mathcal{MC}(C_n)=(2\lfloor n/3\rfloor+(n\mod 3))(\alpha-1)$. The term $2\lfloor n/3\rfloor$ should be trivial according to the definition of $\mathcal{S}_n$ and the additive part $(n\mod 3)$ can be easily attained by a case distinction.

We want to determine $s(n):=|\mathcal{S}_n|$. Let $p(n)$ be $s(n)$ but for a path $P_n$ instead of a cycle $C_n$. We observe that $p(n-2)\le s(n)\le p(n)$. The upper bound is trivial; the lower bound is true because if we choose the first position in $v_{n-1}$ and $v_{n}$ in $C_n$ and let the rest of nodes, i.e., $v_1,\cdots, v_{n-2}$, copy the positions from a solution on a path of length $n-2$, then the outcome is a solution on $C_n$. Thus, to approximate $s(n)$, it suffices to calculate $p(n)$. For that, we need to solve the homogeneous linear difference equation $p(n)=p(n-1)+(\alpha-1)p(n-3)$ for $p(3)= 3\alpha-2, p(4)=\alpha^2+2\alpha-2, p(5)=3\alpha^2-\alpha-1$.
Consider a path of length $n$. If we choose the first position in the first node, then there are $p(n-1)$ ways to choose the positions in the remaining $n-1$ nodes. If we choose any of the other $\alpha-1$ positions in the first node, the second and third nodes are forced to pick the first position and the positions in the remaining $n-3$ nodes can be selected in $p(n-3)$ ways. This is why the above equation holds. It is straightforward to calculate the initial conditions $p(3)$, $p(4)$, and $p(5)$.

By solving the characteristic equation $\lambda^3-\lambda^2-(\alpha-1)=0$ (derived from $p(n)=p(n-1)+(\alpha-1)p(n-3)$) and some standard calculations, we get $p(n)=\Tilde{\Theta}(\psi^n)$ for some $\psi\in (\alpha^{1/3}, \alpha^{1/3}+0.22)$, where $\psi$ converges to $\alpha^{1/3}$ when $\alpha$ grows. (The details are left out; please see~\cite{greene1990mathematics} for more details on solving linear difference equations.) Finally, since $p(n-2)\le s(n)\le p(n)$ and we hide polynomial terms in $\alpha$ in $\Tilde{\Theta}$, we conclude that $s(n)=\Tilde{\Theta}(\psi^n)$.

\subsection{Proof Sketch of Part III in Theorem~\ref{rc-thm}}
\label{rc-appendix}

For a graph $G$, assume that initially each node has a token. Then, in each round, the tokens in a node $v$ update their location according to the following rule: w.p. $1-q$, for some constant $0<q<1$, they all stay in $v$ and w.p. $q$, they pick one of $v$'s neighbors u.a.r. and all move there. The \emph{meeting time} of this process, called \emph{tokens random walk}, is the expected number of rounds required for all tokens to arrive at the same node.

We claim that the convergence time of the Random PD model on a graph $G$ is upper-bounded by the meeting time of the tokens random walk process on $G$. For that, we need to look at the Random PD process in the reverse order. Let's assume that each node has a token in round $t$. Then, for each node we move its token to the node that it got its order from in the round $t$ (it could be the node itself). Now, for each node which has some tokens, we move them to the node that it got its order from in the round $t-1$ and we continue this procedure. Since each node chooses its own order w.p. $1-q$ and otherwise picks the order of one of its neighbors u.a.r., this process is identical to the tokens random walk process. (Please refer to~\cite{lovasz1993random} for more details on reversible Markov chains.) Furthermore, if the token which originated from a node $v$ is located on a node $u$ in some round $t'$, then it means that the final order of $v$ is equal to the order of $u$ in round $t'$. If all tokens reach the same node, the final order of all nodes is equal to the order of that node. This is why the convergence time of the Random PD process is bounded by the meeting time of the tokens random walk process.

According to~\cite{coppersmith1993collisions}, we know that the meeting time can be bounded by $\mathcal{O}(n^4)$. This bound is actually proven in a more restrictive set-up, where a demon, whose goal is to maximize the meeting time, selects the token(s) which move in each round. Furthermore, they conjecture that the bound can be improved to $\mathcal{O}(n^3)$. Finally, it is worth emphasizing that for special classes of graphs better bounds are known. For example, on graphs with strong expansion properties, the meeting time is in $\mathcal{O}(n\log n)$, cf.~\cite{cooper2010multiple}.

\end{document}